\documentclass[11pt]{article}
\usepackage[top=1in,bottom=1in,left=1in,right=1in]{geometry}            
\usepackage[authoryear]{natbib}

\usepackage[all]{xy}
\usepackage{amssymb,latexsym,amsfonts, amsthm, amsmath}
\usepackage{fancyhdr}
\usepackage{verbatim}
\usepackage{enumitem}
\usepackage{youngtab}
\usepackage{multicol}
\usepackage{array}
\usepackage{relsize}
\usepackage{marginnote}

\usepackage{tikz}
\usetikzlibrary{arrows,chains,matrix,positioning,scopes}
\makeatletter
\tikzset{join/.code=\tikzset{after node path={%
\ifx\tikzchainprevious\pgfutil@empty\else(\tikzchainprevious)%
edge[every join]#1(\tikzchaincurrent)\fi}}}
\tikzset{>=stealth',every on chain/.append style={join},
         every join/.style={->}}

\usepackage{breqn}

 \usepackage{xr}
\newtheorem{thm}{Theorem}[section]   
\newtheorem{lemma}[thm]{Lemma}         
\newtheorem{prop}[thm]{Proposition}  

\theoremstyle{definition}

\newtheorem{ex}[thm]{Example}        

\def\bfmath#1{\mathchoice
        {\mbox{\boldmath$#1$}}%
        {\mbox{\boldmath$#1$}}%
        {\mbox{\boldmath$\scriptstyle#1$}}%
        {\mbox{\boldmath$\scriptscriptstyle#1$}}}%
\def\sfsvec{\ensuremath{\bfmath{\xi}}}
\newcommand{\spex}[3]{{#1}^{\raisebox{#2}{\scriptsize$#3$}}}
\newcommand{\join}{\mathcal{J}}

\newcommand{\cone}{\operatorname{cone}}

\newcommand{\RR}{\ensuremath{\mathbb{R}}}
\newcommand{\PP}{\ensuremath{\mathbb{P}}}

\newcommand{\C}{\ensuremath{\mathcal{C}}}

\newcommand{\sref}[1]{Section~\ref{#1}}
\newcommand{\propref}[1]{Proposition~\ref{#1}}
\newcommand{\thmref}[1]{Theorem~\ref{#1}}
\newcommand{\fref}[1]{Figure~\ref{#1}}

\DeclareMathOperator*{\argmin}{arg\,min}
\newcommand{\dadi}{\ensuremath{\partial\texttt{a}\partial\texttt{i}}}
\newcommand{\fastNeutrino}{\texttt{fastNeutrino}}
\newcommand{\fastsimcoal}{\texttt{fastsimcoal2}}
\newcommand{\msprime}{\texttt{msprime}}
\usepackage{authblk,url}

\title{Geometry of the sample frequency spectrum and\\ the perils of demographic inference}
\author[a,1]{Zvi Rosen}
\author[b,c,1]{Anand Bhaskar}
\author[d]{Sebastien Roch}
\author[a,e,f,2]{Yun S. Song}
\affil[a]{Department of Statistics, University of California, Berkeley, CA 94720, USA}
\affil[b]{Department of Genetics, Stanford University, Stanford, CA 94305, USA}
\affil[c]{Howard Hughes Medical Institute, Stanford University, Stanford, CA 94305, USA}
\affil[d]{Department of Mathematics, University of Wisconsin, Madison, WI 53706, USA}
\affil[e]{Computer Science Division, University of California, Berkeley, CA 94720, USA}
\affil[f]{Chan Zuckerberg Biohub, San Francisco, CA 94158, USA}
\affil[1]{Z.R. and A.B. contributed equally to this work}
\affil[2]{To whom correspondence should be addressed: yss@berkeley.edu}
\date{
\today
}

\begin{document}
\maketitle

\renewcommand{\thefootnote}{\fnsymbol{footnote}}

\begin{abstract} 
The sample frequency spectrum (SFS), which describes the distribution of mutant alleles in a sample of DNA sequences, is a widely used summary statistic in population genetics.  The expected SFS has a strong dependence on the historical population demography and this property is exploited by popular statistical methods to infer complex demographic histories from DNA sequence data.   Most, if not all, of these inference methods exhibit pathological behavior, however.
Specifically, they often display runaway behavior in optimization, where the inferred population sizes and epoch durations can degenerate to 0 or diverge to infinity, and show undesirable sensitivity of the inferred demography to perturbations in the data.  The goal of this paper is to provide theoretical insights into why such problems arise.  To this end, we characterize the geometry of the expected SFS for piecewise-constant demographic histories and use our results to show that the aforementioned pathological behavior of popular inference methods is intrinsic to the geometry of the expected SFS.
We provide explicit descriptions and visualizations for a toy model with sample size 4, and generalize our intuition to arbitrary sample sizes $n$ using tools from convex and algebraic geometry.  We also develop a universal characterization result which shows that the expected SFS of a sample of size $n$ under an \emph{arbitrary} population history can be recapitulated by a piecewise-constant demography with only $\kappa_n$ epochs, where $\kappa_n$ is between $n/2$ and $2n-1$. The set of expected SFS for piecewise-constant demographies with fewer than $\kappa_n$ epochs is open and non-convex, which causes the above phenomena for inference from data.
\end{abstract}

\clearpage
\section{Introduction}

\label{sec:intro}

The sample frequency spectrum (SFS), also
known as the site or allele frequency spectrum, is a fundamental statistic in
population genomics for summarizing the genetic variation in a sample of DNA sequences. Given
a sample of $n$ sequences from a panmictic (i.e., randomly mating) population, the SFS is a vector of
length $n-1$ of which the $k$th entry corresponds to the number of segregating sites each with $k$ mutant (or derived) alleles and $n-k$ ancestral alleles.
The SFS provides a compact way to summarize $n$ sequences of arbitrary length into just $n-1$ numbers, and is frequently
used in empirical population genetic studies to test for deviations from
equilibrium models of evolution. For instance, the SFS has been widely used to infer demographic
history where the effective population size has changed over time
\citep{nielsen2000estimation,gutenkunst:2009,gravel2011demographic,keinan:2012,excoffier:2013,bhaskar:2015}, and to test for selective neutrality
\citep{kaplan:1989,achaz:2009}.  In fact, many commonly used
population genetic statistics for testing neutrality, such as Watterson's $\theta_{W}$
\citep{watterson:1975}, Tajima's $\theta_{\pi}$ \citep{tajima:1983}, and Fu and
Li's $\theta_{FL}$ \citep{fu:1993} can be expressed as linear functions
of the SFS \citep{durrett:2008}.

In the coalescent framework \citep{kin:1982:EPS,kin:1982:JAP,kin:1982:SPA},  the {\em unnormalized expected} SFS $\sfsvec_n$ 
for a random sample of $n$ genomes drawn from a population
is obtained by taking the expectation of the SFS over the distribution of sample genealogical histories under a specified population demography.  
In this work, we will be concerned with well-mixed, panmictic populations with time-varying historical population sizes,
evolving according to the neutral coalescent process with the infinite-sites model of mutation. 
The coalescent arises as the continuum limit of a large class of
discrete models of random mating, such as the Wright-Fisher, Moran, and Cannings
exchangeable family of models \citep{mohle:2001}. The infinite-sites model
postulates that every mutation in the genealogy of a sample occurs at a distinct
site, and is commonly employed in population genetic studies for organisms with
low population-scaled mutation rates, such as humans.  The SFS also appears in
the context of statistical modeling as a vector of probabilities.  In
particular, the {\em normalized expected} SFS $\widehat{\sfsvec}_{n}$, defined by
normalizing the entries of $\sfsvec_n$ so that they sum to 1, gives
the probability that a random mutation appears in $k$ out of $n$ sequences in the sample.
Unless stated otherwise, we use the term expected SFS to refer to the
unnormalized quantity $\sfsvec_n$.

The expected SFS is strongly influenced by
the demographic history of the population, and extensive theoretical and
empirical work has been done to characterize this dependence \citep{Fu,wakeley1997estimating,PBK,marth:2004,chen2012joint,kamm2017efficient,jouganous2017inferring}.
\citet{Fu} showed that under the infinite-sites model for a panmictic population
with constant size and no selection, the expected SFS is given by $\sfsvec_n=\theta\cdot\big(1,\frac{1}{2},\ldots,\frac{1}{n-1}\big)$, where $\theta/2$ denotes the population-scaled 
mutation rate. When the population size is variable, however, the formula
for the expected SFS depends on the entire population size history. In particular,
\citet{PK} showed that the expected SFS under a time-varying population size is
given by $\sfsvec_n = A_n \mathbf{c}$, with $A_n$ being an $(n-1)$-by-$(n-1)$ invertible matrix
that only depends on $n$ and $\mathbf{c}=(c_2,\ldots, c_n)$, where
$c_m$ denotes the expected time to the first coalescence event in a random
sample of size $m$ drawn from the population at present. 
For any time-varying population size function $\eta(t)$, $c_m$ is given by the following expression:
\begin{equation}
c_m = \int_0^\infty \binom{m}{2} \frac{1}{\eta(t)} \exp \left[
-\binom{m}{2} \int_0^t \frac{1}{\eta(x)} \mathrm{d}x \right] \mathrm{d}t. 
\label{eq:c_m}
\end{equation}

A natural statistical question that arises when using the SFS for demographic
inference is whether it is theoretically possible to reconstruct the population history $\eta(t)$ from the
expected SFS $\sfsvec_n$ of a large enough sample size $n$. This question was famously answered
by \citet{MFP} in the negative, by constructing a way of perturbing any
population size history without altering the expected SFS for all sample sizes. However, in subsequent work by two of this paper's authors \citep{bhaskar2014descartes}, it
was shown that for a wide class of biologically plausible population histories,
such as those given by piecewise-constant and piecewise-exponential functions,
the expected SFS of a finite sample size is sufficient to uniquely identify the
population history. These two results give us insight into the map from
population size histories to the expected SFS vectors. The space of {\em all}
possible population size histories is of infinite dimension, while the expected
SFS vectors for any fixed sample size $n$ form a finite-dimensional space;
naturally, the pre-image of an expected SFS under this map will typically be
an infinite set of population size histories. However, imposing conditions such
as those of \citet{bhaskar2014descartes} restricts us to a function space 
of finite dimension, so that the pre-images of
the expected SFS may become finite/unique.

While the above results are concerned with the identifiability of demographic
models from noiseless SFS data, they do not directly provide an explicit
characterization of the geometry of the expected SFS as a function of the
demographic  model.  Studying such geometry would be very useful for
understanding the  behavior of inference algorithms which
perform optimization by repeatedly computing the image of the map from the
space of demographic parameters to the expected SFS while 
trying to minimize the deviation of the expected SFS from the observed SFS data.
To this end, our main contribution is a universal characterization of the space of expected SFS for any
sample size $n$ under \emph{arbitrary population size histories}, in terms of the
space of expected SFS of piecewise-constant population size functions with $O(n)$ epochs.  This is a useful reduction because the
latter space is much more tractable for mathematical analysis and computation.
We provide a complete geometric description of this space for a sample of size
$n=4$, and generalize our intuition to arbitrary sample sizes $n$ using tools from convex and algebraic geometry.
Our reduction also provides an explanation for a puzzling phenomenon frequently
observed in empirical demographic inference studies -- namely, for some observed
SFS data, the optimization procedure for inferring population histories 
sometimes exhibits 
pathological behavior where the inferred population sizes and epoch durations can
degenerate to 0 or diverge to $\infty$.

\section{Piecewise-Constant Demographies}
\label{sec:piecewise_constant}
Let $\Pi_k$ be the set of piecewise-constant population size functions with $k$
pieces. Any population size function in $\Pi_k$ is described by $2k-1$ positive
numbers, representing the $k$ population sizes $(y_1,\ldots,y_k)$ and the $k-1$ time points 
$(t_1,\ldots,t_{k-1})$ when
the population size changes. Let $\Xi_{n,k}$, which we call the $(n,k)$-SFS manifold\footnote{
The sets $\Xi_{n,k}$ and $\C_{n,k}$ are not technically manifolds; they would be more accurately described as semialgebraic sets. However, for expository purposes, we use the widely known term ``manifold.''
},
denote the set of all expected SFS vectors for a sample of size $n$ that can be
generated by population size functions in $\Pi_k$. Similarly, let $\C_{n,k}$,
called the $(n,k)$-coalescence manifold, denote the set of all vectors
$\mathbf{c}=(c_2,\ldots,c_n)$ giving the expected first coalescence times of samples of
size $2, \ldots, n$ for population size functions in $\Pi_k$. Let
$\widehat{\Xi}_{n,k}$ and $\widehat{\C}_{n,k}$ respectively be equal to the normalization of
all points in $\Xi_{n,k}$ and $\C_{n,k}$ by their $\ell_1$-norms (i.e., the sums of
their coordinates). Note that both manifolds live in $\RR^{n-1}$
and their normalized versions live in the $(n-2)$-dimensional simplex
$\Delta^{n-2}$; this is the set of nonnegative vectors in $\RR^{n-1}$ whose
coordinates sum to $1$.

Now that we have defined our basic objects of study, we can describe the remainder
of the paper:
In \sref{sec:4k}, we provide a complete geometric picture of the $\Xi_{4,k}$ SFS manifold
describing the expected SFS for samples of size $n=4$ under piecewise-constant
population size functions with an arbitrary number $k$ of pieces. We make explicit
the map between regions of the demographic model space and the corresponding probability
vectors, and this will foreshadow some of the difficulties with population size inference
in practice.
In \sref{sec:nk}, we develop a characterization of the space of expected
SFS for arbitrary population size histories.  In particular, we show that for
any sample size $n$, there is a finite integer $\kappa_n$ such that the expected
SFS for a sample of $n$ under any population size history can be generated by a
piecewise-constant population size function with at most $\kappa_n$ epochs.
Stated another way, we show that the $\Xi_{n, \kappa_n}$ SFS manifold contains the
expected SFS for all possible population size histories, no matter how
complicated their functional forms. We establish bounds on $\kappa_n$ that are
linear in $n$, and along the way prove some interesting results regarding the
geometry of the general $\Xi_{n,k}$ SFS manifold.
Finally, in \sref{sec:ml_inference}, we demonstrate the implications of our
geometric characterization of the expected SFS for the problem of demographic
inference from noisy genomic sequence data.

Before proceeding further, we state a proposition regarding the structure of the
map from $\Pi_k$ to $\C_{n,k}$, which we will call $\chi(\vec{x},\vec{y})$; the
vector of  $k-1$ transformed breakpoints is denoted by $\vec{x}=(x_1,\ldots,x_{k-1})$ and defined below, while the vector of population sizes in the $k$ epochs
is denoted by $\vec{y}=(y_1,\ldots,y_k)$. This allows us to explain the algebraic nature of most of our proofs. 
All proofs of the results presented in this paper are deferred to \sref{sec:proofs}.

\begin{prop}
Fix a piecewise-constant population size function in $\Pi_k$ with epochs
$[t_0,t_1)$, $[t_1,t_2),\ldots$, $[t_{k-1},t_k)$, where $0 = t_0 < t_1 < \cdots
< t_{k-1} < t_k = \infty$, and which has constant population size value $y_j$ in the
epoch $[t_{j-1},t_j)$ for $j=1,\ldots,k$.  Let $x_j = \exp[-(t_{j}- t_{j-1})/y_j]$ for $j=1,\ldots,k$, where $x_k = 0$ (corresponding to time
$T = \infty$), and define $x_0 = 1$ (corresponding to time $T = 0$) for
convenience.  The vectors $(x_1,\ldots,x_{k-1},y_1,\ldots,y_k)$, where $0 <
x_j < 1$ and $y_j > 0$ for all $j$, (uniquely) identify the population size functions 
in $\Pi_k$, and they satisfy both of the following equations:

\begin{align}
 \left[ \arraycolsep=1pt \begin{array}{cccc}
x_0(1-x_1) & x_0x_1(1-x_2) & \ldots & \left(\prod_{i=0}^{k-1} x_i\right) ( 1- x_k) \\[3mm]
\frac{1}{3}x_0^3(1-x_1^3) &\frac{1}{3} x_0^3x_1^3(1-x_2^3) & \ldots & \frac{1}{3}\left(\prod_{i=0}^{k-1} x_i^3\right) ( 1- x_k^3) \\
\vdots & \vdots & \ddots & \vdots \\
\medmuskip=-1mu
\thinmuskip=-1mu
\thickmuskip=-1mu
\nulldelimiterspace=0pt
\scriptspace=0pt
{\scriptstyle \frac{1}{{\binom{n}{2}}}x_0^{\binom{n}{2}}\Big(1-x_1^{\binom{n}{2}}\Big)} \hspace*{3mm} &
{\scriptstyle \frac{1}{{\binom{n}{2}}}x_0^{\binom{n}{2}}x_1^{\binom{n}{2}}\Big(1-x_2^{\binom{n}{2}}\Big)} &
\ldots &
{\scriptstyle \frac{1}{{\binom{n}{2}}}\Big(\prod_{i=0}^{k-1} x_i^{\binom{n}{2}}\Big) \Big( 1- x_k^{\binom{n}{2}}\Big)} 
\end{array} \right] \left[ \begin{array}{c}
y_1 \\ y_2\\ \vdots \\ y_k \end{array} \right] &  =  
\left[ \begin{array}{c} c_2 \\ c_3\\ \vdots \\ c_n \end{array} \right] &
\label{formula1}
\\
\text{ and } \hspace{1cm} \left[\begin{array}{cccc} 1 & x_1 & \ldots & \prod_{i=1}^{k-1} x_i  \\
\frac{1}{3} & \frac{1}{3} x_1^3 & \cdots & \frac{1}{3} \prod_{i=1}^{k-1} x_i^3  \\
\vdots & \vdots & \ddots & \vdots \\
 \frac{1}{\binom{n}{2}} & \frac{1}{\binom{n}{2}} x_1^{\binom{n}{2}} & \cdots &  \frac{1}{\binom{n}{2}} \prod_{i=1}^{k-1} x_i^{\binom{n}{2}} \\
\end{array} \right]
\left[\begin{array}{c} y_1 \\ y_2 - y_1 \\ \vdots \\ y_k - y_{k-1} \end{array} \right] & = 
\left[ \begin{array}{c} c_2 \\ c_3\\ \vdots \\ c_n \end{array} \right], &
\label{formula2}
\end{align}
where $c_m$ is the expected first coalescence time for a sample of size $m$, as defined in \eqref{eq:c_m}.
\label{prop:2formulas}
\end{prop}

\noindent
These two formulations provide two perspectives on the coalescence manifold $\C_{n,k}$:
\begin{enumerate}
\item In \eqref{formula1}, the left-hand matrix, call it $M_1(n,k)$, has each column of the same form with two parameters;
this indicates they all live in a 2-dimensional surface.
The vector $(y_1,\ldots,y_k)$ has all positive parameters.
This means that the vector $\mathbf{c}=(c_2,\ldots,c_n)$ is contained
in the cone over the surface described by the columns of $M_1$.

\item  In \eqref{formula2}, the left-hand matrix, call it $M_2(n,k)$
has each column of the same form with one parameter;
this indicates they all live on a curve. The vector $(y_1,y_2-y_1,\ldots,y_k-y_{k-1})$ on the left hand side has
parameters with possibly negative coordinates.
So the vector $\mathbf{c}=(c_2,\ldots,c_n)$ is contained in the linear span of the curve
described by the columns of $M_2$.
\end{enumerate}

\propref{prop:2formulas} gives us the algebraic mappings that will
serve as our objects of interest. Since the SFS manifold is
simply a linear transformation of the coalescence manifold,
we will use these maps as our entry into understanding the SFS manifold.

\section{The $\Xi_{4,k}$ SFS Manifold: A Toy Model}
\label{sec:4k}

The first in-depth study will involve the set of all possible
expected SFS for a sample of size $4$. We choose $n=4$
for a number of reasons: First, the cases of $n = 2$ and $3$ are
cones with simple boundaries in the line or plane.
Second, when $n=4$, the absolute SFS manifold lives in $\RR^3$,
which can be nicely visualized, and the normalized SFS manifold lives
in the $2$-simplex, i.e. the triangle with vertices $(1,0,0), (0,1,0)$,
and $(0,0,1)$. Finally, as observed in \propref{prop:2formulas},
the most interesting phenomena in SFS manifolds of any dimension are
fundamentally phenomena of curves and surfaces. These are already captured
in the $n=4$ case.

For the sake of completeness, we begin by formally describing the coalescence
manifolds $\C_{n,k}$ for the trivial cases of $n=2$ and $n=3$.

\begin{prop}
We list some basic results on the coalescence manifolds $\C_{n,k}$ for small values of $(n,k)$:
\begin{enumerate}
\item 
$\C_{n,1}   =  \bigg\{ \lambda \cdot\bigg(1,\dfrac{1}{3},\ldots,\dfrac{1}{\binom{n}{2}} \bigg) :  \lambda  > 0 \bigg\},$  for all $n$.
\item
$\C_{2,k} = \C_{2,1} =  \{ a  :   a  > 0\}$, for all $k \geq 1$.
\item 
$\C_{3,k} = \C_{3,2} =\{ (a,b) :  a > 0 \: \text{  and  } \:  0 < b < a\}$,   for all $k \geq 2.$
\end{enumerate}
\label{prop:trivial}
\end{prop}

Note that from \eqref{formula1} and \eqref{formula2}
for $\chi(\vec{x},\vec{y})$ (\sref{sec:piecewise_constant}), 
it follows that $\chi(\vec{x},a\vec{y})= a\chi(\vec{x},\vec{y})$ for $a > 0$. 
In words, rescaling the population sizes in each epoch by a constant $a$ also rescales the first coalescence times by $a$.
This implies that every
point in the coalescence manifold $\C_{n,k}$
generates a full ray contained in the $\C_{n,k}$
coalescence manifold. Another consequence
is that the normalized coalescence manifold $\widehat{\C}_{n,k}$ is
precisely the intersection of the
coalescence manifold $\C_{n,k}$ with the simplex $\Delta^{n-2}$.

With that justification, we begin to consider the normalized coalescence manifold
$\widehat{\C}_{4,k}$ living in the simplex.
As stated in \propref{prop:trivial}, $\C_{4,1}$ is a ray, which implies
that $\widehat{\C}_{4,1}$ is a single point.
We now characterize the set $\widehat{\C}_{4,2}$.

\begin{figure}[t]
\centerline{\includegraphics[width=\textwidth]{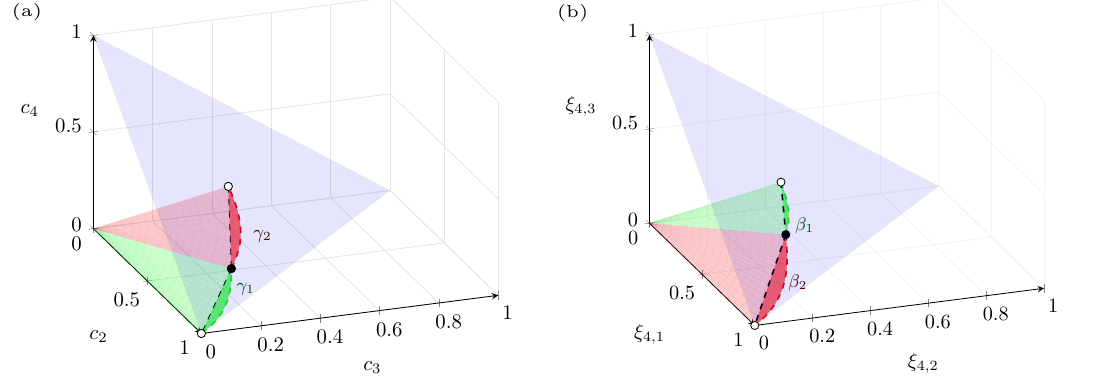}}

\caption{{\bf Coalescence and SFS manifolds for sample size $4$ and
$2$ population epochs}. (a) The coalescence manifold $\C_{4,2}$ is the union of
red and green cones. The 2-simplex, shaded in blue, intersects $\C_{4,2}$
in the normalized coalescence manifold $\widehat{\C}_{4,2}$. The green region
corresponds to small-then-large demographies;
the red region to large-then-small demographies. (b) The SFS manifold
$\Xi_{4,2}$ is the union of
red and green cones. The 2-simplex intersects $\Xi_{4,2}$
in the normalized SFS manifold $\widehat{\Xi}_{4,2}$. Here,
too, the green region
corresponds to small-then-large demographies;
the red region to large-then-small demographies. $\Xi_{4,2}$ is obtained
from $\C_{4,2}$ by a linear transformation. }
\label{fig:42}
\end{figure}

\begin{prop}
The manifold $\widehat{\C}_{4,2}$ is a two-dimensional subset of the 2-simplex
which can be described as the union of the point $\widehat{\C}_{4,1}$
with the interiors of the convex hulls of two curves
$\gamma_1$ and $\gamma_2$. The curves are parametrized as follows:
\begin{align*}
\gamma_1 & =  \left\{ \left(\dfrac{6}{6 + 2t^2 + t^5} ,
\dfrac{2t^2}{6 + 2t^2 + t^5},  \dfrac{t^5}{6 + 2t^2 + t^5}\right)
 :  0<t<1 \right\}, \\
\text{ and    } \hspace{5mm} \gamma_2 & =   \left\{
\left( \dfrac{6}{6 + 2[2]_t + [5]_t},
\dfrac{2[2]_t}{6 + 2[2]_t + [5]_t},
\dfrac{[5]_t}{6 + 2[2]_t + [5]_t} \right)  :  0<t<1 \right\},  
\end{align*}
where $[n]_t$ denotes $1 + \cdots + t^n$.
\label{prop:c42}
\end{prop}

This set has some highly unpleasant geometry. First of all, the set is non-convex; it is also
neither closed nor open, because most of the boundary
is excluded with the exception of the point $(2/3,2/9,1/9)$.
The set is visualized in Figure~\ref{fig:42}(a).

In order to precisely illustrate the geometry of $\chi(\vec{x},\vec{y})$,
we will consider how contours in the domain map to contours in the image.
Specifically, we plot the images of lines with fixed values of $x_1$, respectively
fixed values of $(y_1,y_2)$, to $\C_{4,2}$ in the $2$-simplex. The resulting
contours are pictured in Figure~\ref{fig:contours}.

\begin{figure}[t]
\centerline{
\begin{tabular}{cc}
(a) \raisebox{-.9\height}{\includegraphics[width=.35\textwidth]{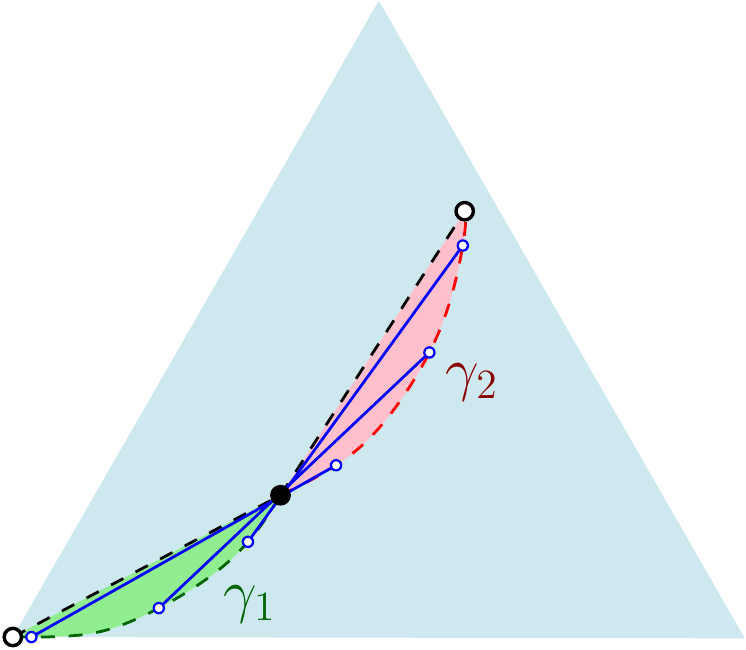}} &
(b) \raisebox{-.9\height}{\includegraphics[width=.35\textwidth]{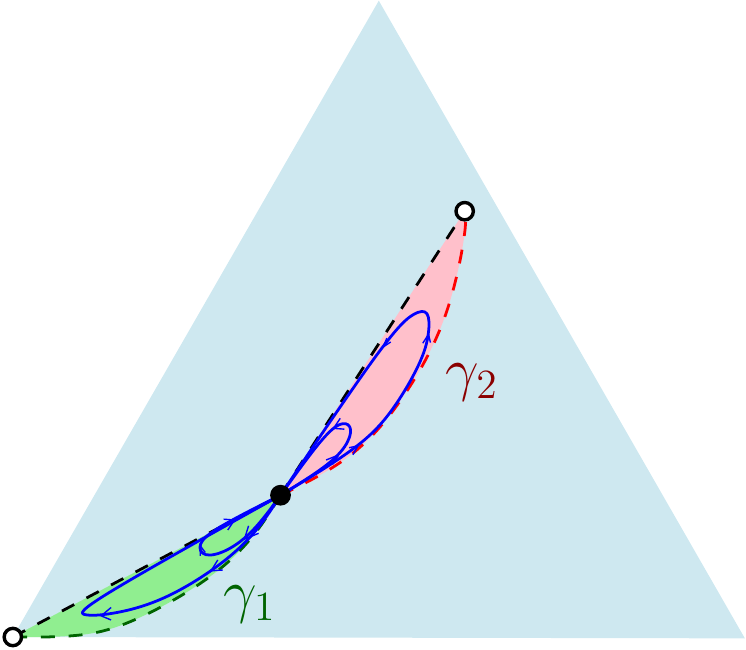}}
\end{tabular}
}
\caption{{\bf Fixed-time and fixed-size contours
in $\widehat{\C}_{4,2}$.} (a) The blue line segments correspond
to the image of $\chi_{4,2}(x^*,\vec{y})$ where $x^*$ is a constant fixing
the break-point between the two demographies.
The other input $\vec{y} = (y_1,y_2)$
varies over all positive vectors, though scaled $\vec{y}$ vectors point to the
same normalized value.
As $y_1/y_2 \to 0$, the image
approaches $\gamma_2$ and as $y_2/y_1 \to 0$, the image approaches
$\gamma_1$.  (b) The blue curves  correspond
to the image of $\chi_{4,2}(x,\vec{y}^*)$ where $\vec{y}^*$ is a fixed
vector indicating the population values
and $x$ takes all values in $(0,1)$. The endpoints $0$ and $1$ correspond
to breakpoints at $\infty$ and $0$ respectively. For $y_1^*< y_2^*$,
$x$ traces a loop in the green region; for $y_1^*>y_2^*$, $x$ traces a loop
in the red region.
}
\label{fig:contours}
\end{figure}

\begin{figure}[t]
\centerline{\includegraphics[width=\textwidth]{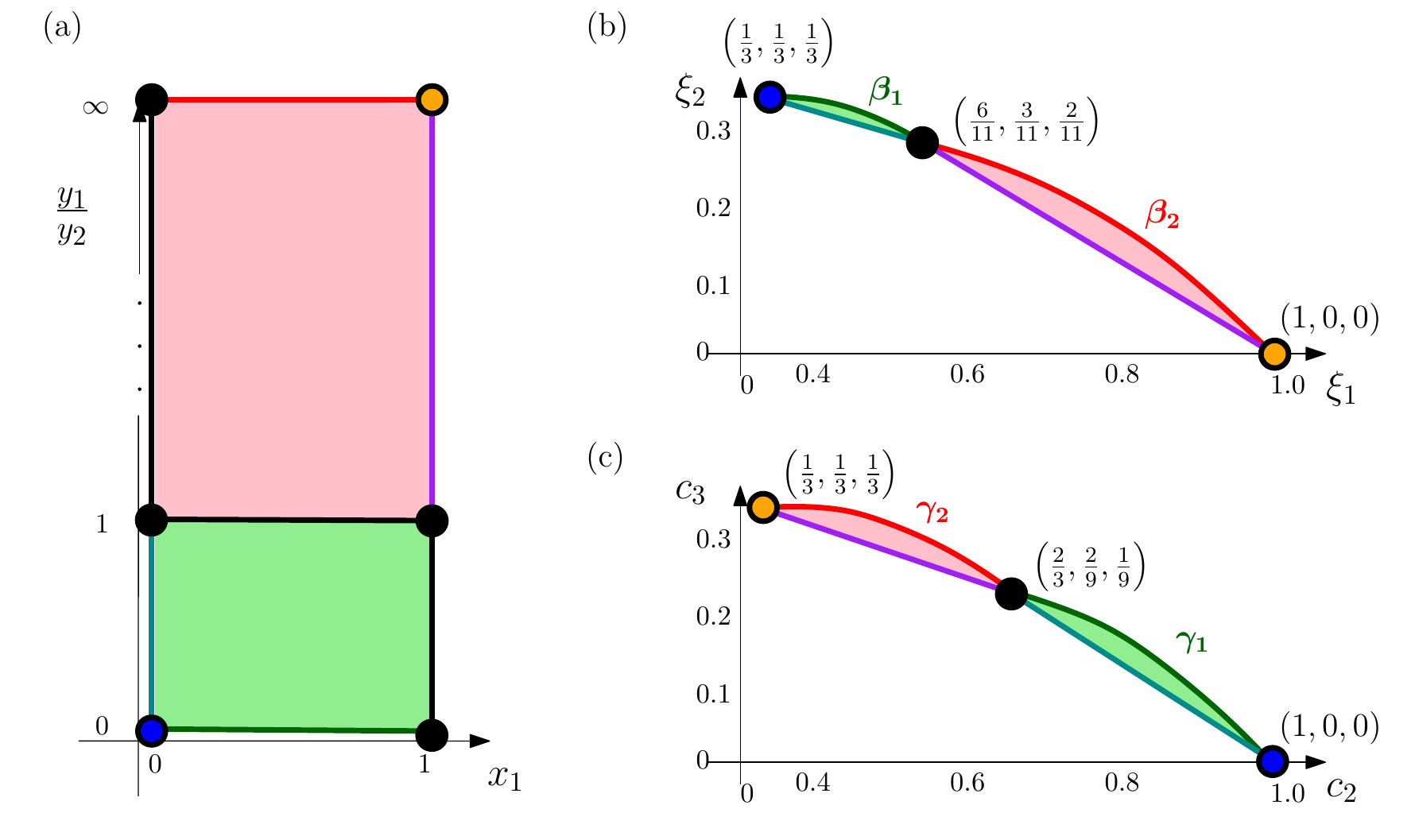}} 
\caption{{\bf Pairing the boundaries of demography space and $\widehat{\C}_{4,2}$}.
(a) The domain of $\chi_{4,2}$.  Note that
for fixed $y_1/y_2$, the normalized coalescence vector is the same.
(b) The normalized SFS manifold $\widehat{\Xi}_{4,2}$
projected onto its first two coordinates. 
(c) The normalized coalescence manifold $\widehat{\C}_{4,2}$ projected onto
its first two coordinates.
The red square at left corresponding to $y_1 > y_2$
maps to the red regions at right; the green square at left corresponding
to $y_2 < y_1$ maps to the green regions at right.
The black line segments on left (corresponding to $y_1/y_2 = 1$;
$y_2 < y_1$ and $x_1 = 0$ (equivalently $t_1 = \infty$); 
$y_2 > y_1$ and $x_1 = 1$ (equivalently $t_1 = 0$))
 all map to the central black points on right, since they
each mimic a constant demography. 
The green line corresponding to $y_1 = 0$
maps to the curve $\beta_1$ in $\widehat{\Xi}_{4,2}$ and the curve $\gamma_1$ in $\widehat{\C}_{4,2}$; 
the red line corresponding to
$y_2 = 0$ maps to the curve $\beta_2$ in $\widehat{\Xi}_{4,2}$ and the curve $\gamma_2$ in $\widehat{\C}_{4,2}$. The orange point
$(x_1 = 1, y_2 = 0)$ maps to $(\frac{1}{3},\frac{1}{3},\frac{1}{3})$ in $\widehat{\C}_{4,2}$
and maps to $(1,0,0)$ in $\widehat{\Xi}_{4,2}$.
The blue point $(x_1 = 0, y_1 = 0)$ maps to $(1,0,0)$ in $\widehat{\C}_{4,2}$
and $(\frac{1}{3},\frac{1}{3},\frac{1}{3})$ in $\widehat{\Xi}_{4,2}$.
The remaining aqua and violet segments
map to the segments of the same color.
}
\label{fig:boundaries}
\end{figure}

Finally, we consider how the map $\chi$ acts on the boundaries of
the domain. To aid visualization, we limit the inputs to
$x_1$ and $y_1/y_2$, since all rescalings of $y_1$ and $y_2$ by the same positive constant while keeping $x_1$ fixed  map to the same normalized coalescence vector.
The resulting map is illustrated in Figure~\ref{fig:boundaries}.

We note that the map fails to be one-to-one within the domain only
when $y_1 = 1$; this is also in the pre-image
of the point $(\frac{2}{3},\frac{2}{9},\frac{1}{9}) \in \widehat{\C}_{4,2}$. The inverse function theorem implies
that on the complement of $y_1 = 1$, the map is a homeomorphism.
This is consistent with our observation that the two rectangles in Figure~\ref{fig:boundaries}(a)
correspond to the two envelopes in Figure~\ref{fig:boundaries}(c).

\begin{prop}
For all values $k \geq 3$, the manifold $\widehat{\C}_{4,k} = \widehat{\C}_{4,3}$, and 
$\widehat{\C}_{4,3}$ is the interior of the convex hull of the following curve:
\begin{align*}
\gamma_3 = \left\{ \left(\dfrac{1}{1 + t^2 + t^5} ,
\dfrac{t^2}{1 + t^2 + t^5},  \dfrac{t^5}{1 + t^2 + t^5}\right)
 :  0<t<1 \right\}.
\end{align*}
\label{prop:c43}
\end{prop}

As we can see from \propref{prop:c43}, $\widehat{\C}_{4,3}$ is open and convex;
however, we lose one useful property of the normalized map $\widehat{\chi}:\RR^3 \to \widehat{\C}_{4,2}$.  
Specifically, let  $\widehat{\chi}':\RR^2 \to \widehat{\C}_{4,2}$ be given by
$\widehat{\chi}'(x_1,y_1) = \widehat{\chi}(x_1,y_1,1)$, noting
that $\widehat{\chi}(x_1, \lambda y_1, \lambda y_2) = \widehat{\chi}(x_1,y_1,y_2)$ for $\lambda >0$. 
Under this definition $\widehat{\chi}'$ is generically one-to-one. Meanwhile, the 
analogous construction $\widehat{\chi}':\RR^4 \to \widehat{\C}_{4,3}$ mapping
the three-epoch demography with breakpoints $(x_1,x_2)$ and population sizes
$(y_1,y_2,1)$ to the corresponding normalized coalescence vector
has two-dimensional pre-images, generically.
For this reason, contour images do not lend themselves to easy description.
Still, we can at least describe the image of the map on the boundaries of our domain.

The easiest way to visualize the map is first to understand
how the time variables affect the value of the columns of $M_1(4,3)$ and
to view the $y$ variables as specifying points in the convex hull of those
$3$ columns. The boundaries of the square $(x_1,x_2) \in [0,1]\times [0,1]$
map the columns (after rescaling to the simplex) as follows:

\[
\begin{array}{lclclcl}
x_1 = 0 & \mapsto & \left[\begin{array}{c|c|c} 6/9 & 1 & 1 \\ 2/9 & 0 & 0 \\ 1/9 & 0 & 0 \\ \end{array} \right], & \hspace{1cm} &
x_1 = 1 & \mapsto & \left[\begin{array}{c|c|c} 1/3  & |  & | \\ 1/3  & \gamma_2(x_2) & \gamma_1(x_2) \\ 1/3 & | & | \\ \end{array}\right], \\[1cm]
x_2 = 0 & \mapsto & \left[\begin{array}{c|c|c}  |  & | & 1 \\  \gamma_2(x_1) & \gamma_1(x_1) & 0 \\ | & | & 0 \\ \end{array}\right], & &
x_2 = 1 & \mapsto & \left[\begin{array}{c|c|c}  |  & | & | \\  \gamma_2(x_1) & \gamma_3(x_1) & \gamma_1(x_1) \\ | & | & | \\ \end{array}\right]. \\
\end{array}
\]

The case of $x_2 = 1$ is the most interesting:
 when we fix $y_1 = y_3 = 0$ and
$y_2 = 1$, we obtain the boundary curve $\gamma_3(t)$.
Note that $x_2 = 1$ corresponds
to a second epoch of length $0$. The intuition is
that very short population booms at
the second epoch lead to coalescence vectors close to $\gamma_3$.
The maps encoded by a general column of $M_1(4,k)$ correspond
to the interior of the orange region.
Adding in convex combinations of points
gives the lined region, which is the remainder of $\C_{4,3}$; this is
discussed more rigorously in \sref{sec:proofs}.
When the number of epochs $k$
steps higher, all columns of $M_1(4,k)$ still map to
the same region of the simplex,
so $\C_{4,k}$ will still be contained in this convex hull.
The region $\C_{4,3}$ is depicted in Figure~\ref{fig:43}(a).

\begin{figure}[t]
\centerline{\includegraphics[width=\textwidth]{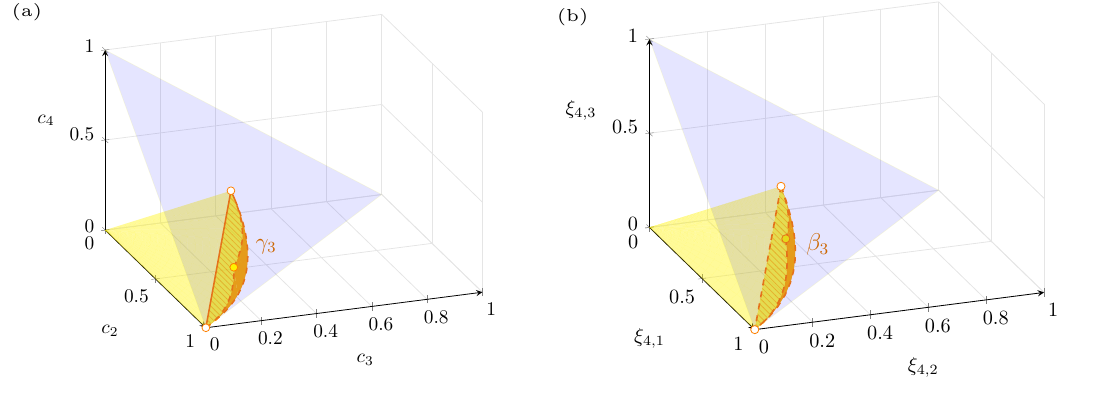}}
\caption{{\bf Coalescence and SFS manifolds for sample size $4$ and
3 population epochs}. (a) The coalescence manifold $\C_{4,3}$ is the entire
yellow and orange region. The 2-simplex, shaded in blue, intersects $\C_{4,2}$
in the normalized coalescence manifold $\widehat{\C}_{4,3}$.
The orange region of $\widehat{\C}_{4,3}$, bounded by
$\gamma_1,\gamma_2$, and $\gamma_3$, is the image of the
surface described by the columns of $M_1(4,3)$, while the yellow region
adds in vectors gained by using linear combinations.
(b) The SFS manifold $\Xi_{4,3}$ is the entire yellow and orange
region. The 2-simplex intersects $\Xi_{4,3}$
in the normalized SFS manifold $\widehat{\Xi}_{4,3}$.
$\Xi_{4,2}$ is obtained from $\C_{4,3}$ by a linear transformation.
The orange region of $\widehat{\Xi}_{4,3}$, bounded by
$\beta_1,\beta_2$, and $\beta_3$, is the image of the
surface described by the columns of $M_1(4,3)$, while the yellow region
adds in vectors gained by using linear combinations.}
\label{fig:43}
\end{figure}

As mentioned earlier, the SFS manifold $\Xi_{n,k}$ is merely a linear transformation
of $\C_{n,k}$; however, since it is of interest in its own right, we include the
formulae for $\Xi_{4,k}$ analogous to those derived in this section.

\begin{prop} The following hold for the normalized $(4,k)$-SFS manifold:
\[ \widehat{\Xi}_{4,1}  =  \left( \dfrac{6}{11}, \dfrac{3}{11},\dfrac{2}{11} \right). \]
$\widehat{\Xi}_{4,2}$ is the union of $\widehat{\Xi}_{4,1}$ with the convex hulls of two curves:
\begin{align*}
\beta_1 & =  \left\{ \left( \dfrac{18 + 10t^2 + 2t^5}{54 + t^5},
 \dfrac{18 - 3t^5}{54 + t^5},
\dfrac{18 -  10t^2 + 2t^5}{54 + t^5} \right) :  0 < t < 1 \right\}, \\
\beta_2 & =  \left\{ \left( \dfrac{18 + 10[2]_t + 2[5]_t}{54 + [5]_t},
\dfrac{18 - 3[5]_t}{54 + [5]_t},
\dfrac{18 - 10[2]_t + 2[5]_t}{54 + [5]_t} \right)  :  0 < t < 1 \right\}.
\end{align*}
Here, also, $[n]_t$ denotes $1 + t + \cdots + t^n$.
Finally, $\widehat{\Xi}_{4,k} = \widehat{\Xi}_{4,3}$ for all $k$, and
$\widehat{\Xi}_{4,3}$ is the convex hull of $\beta_3$, where
\[\beta_3 = \left\{ \left( \dfrac{3 + 5t^2 + 2t^5}{9 + t^5},
 \dfrac{3 - 3t^5}{9 + t^5},
\dfrac{3 - 5t^2 + 2t^5}{9 + t^5} \right): 0 < t < 1 \right\}. \]

\label{prop:sfs}
\end{prop}

Visualizations of $\Xi_{4,2}$ and $\Xi_{4,3}$ may be found in Figure~\ref{fig:42}(b)
and Figure~\ref{fig:43}(b).

\section{The $\Xi_{n,k}$ SFS Manifold: General Properties}
\label{sec:nk}

In this section, we examine the constant $\kappa_n$, defined in \sref{sec:piecewise_constant}
as the smallest index for which $\C_{n,k} \subseteq \C_{n,\kappa_n}$ for all $k$.  The tools
for the proofs in this section come from algebraic geometry (for the derivation of the
lower bound) and convex geometry (for the upper bound).

The gist of the algebraic geometry argument is that, under the $M_2(n,k)$ formulation,
the manifold $\C_{n,k}$ can be seen to be a
relatively open subset of an algebraic variety (manifold)
built by a sequence of well-understood constructions.
Details of this perspective are reserved for the Proofs section.

Two concrete consequences follow from this observation:
\begin{enumerate}
\item the ability to compute all equations satisfied by $\C_{n,k}$
using computer algebra, and
\item a formula for the dimension of the coalescence and SFS manifolds.
\end{enumerate}
While the former is harder to explain without more setup, the latter can
be formulated as follows:

\begin{prop}
The dimension of $\widehat{\C}_{n,k}$
is given by: 
\[\dim {\widehat{\C}_{n,k}} = \begin{cases} 0, & k = 1, \\ \min(2k-2, n-2), & \text{else.} \\  \end{cases} \]
In particular, $\C_{n,k} \subsetneq \C_{n,k+1}$ for $k < \lceil\frac{1}{2}n\rceil$.
\label{prop:cnk}
\end{prop}

We will illustrate how these algebraic ideas can be applied
in the next case we have not seen, namely to the sample size $n=5$.

\begin{ex}
Note that $\widehat{\C}_{5,1} = \left(\frac{30}{48},\frac{10}{48},\frac{5}{48},
\frac{3}{48}\right)$, by \propref{prop:trivial}.
We will use the new ideas above to describe $\widehat{\C}_{5,k}$ for higher values of $k$.

Since the normalized coalescence manifold has dimension $\min(2k-2,n-2)$,
we know that $\widehat{\C}_{5,2}$ has dimension $2$ inside of the $3$-simplex;
therefore, we anticipate that it will satisfy
{\bf one} equation, matching its codimension.
The degree of the algebraic variety implies that this
polynomial should have degree $8$.
Indeed, when we compute this equation using
{\tt Macaulay2} \citep{M2}, we obtain a
huge degree-8 polynomial with 105 terms, whose
largest integer coefficient is $5,598,720$.

Finally, $\widehat{\C}_{5,3}$ is full-dimensional in the $3$-simplex, so it will
satisfy no algebraic equations relative to the simplex.
It would be defined instead by the inequalities
determining its boundary.
\end{ex}
While \propref{prop:cnk} is useful for analyzing 
individual coalescence manifolds,
it also leads to the observation that $\kappa_n 
\geq \lceil\frac{1}{2}n\rceil$, since the inclusions are
proper until that index.
It is worth remarking that a slightly weaker lower bound of $\kappa_n \geq
\lfloor\frac{1}{2}n\rfloor$ follows immediately 
from the identifiability result
of \citet[Corollary 7]{bhaskar2014descartes}, which states that for 
a piecewise-constant population
size function with $k$ pieces, the expected SFS 
of a sample of size $n \geq 2k$ 
suffices to uniquely identify the function.

The convex geometry argument is more elementary. As we noted,
the $M_1$ formulation is contained in the convex hull over
the surface described by a general column of $M_1$.  Because
the columns are related, our selection of points in
the surface is not unrestricted. For this reason, it is not obviously
\emph{equal} to the convex hull. However, once
we fix some collection of values $x_1,\ldots,x_k$
for $\C_{n,k}$, we can use convex geometry for the resulting polytope.
In particular, we use {Caratheodory's Theorem} (\citet{Cara} or
\citet[Theorem 2.3]{Barvinok}), which states that for 
$X$ a subset of $\RR^n$, every $x \in \cone(X)$ can 
be represented as a positive combination of
vectors $x_1,\ldots,x_m \in X$ for some $m \leq n$.

The argument, roughly, allows us to construct any point in that
convex hull, with as few as $n+1$ points. This allows
us to place the point in $\C_{n,j}$ for $j\leq 2n-1$. Since no new
SFS are generated by using more than $2n-1$ epochs, we learn that
 $\kappa_n$ is bounded above by $2n-1$.

Combining the two bounds obtained in this section, we have the
following: 
\begin{thm}
\label{thm:kappa_bounds}
For any integer $n \geq 2$, there exists a positive
integer $\kappa_n$ such that
$\Xi_{n,k} \subseteq \Xi_{n,\kappa_n}$ for all $k \geq 1$. Furthermore, 
$\kappa_n$ satisfies
\[
\lceil n/2 \rceil \: \leq \kappa_n \leq \: 2n -1.
\]
\end{thm}

This allows us to express the SFS from any piecewise-constant 
demography as coming from a demography with relatively few epochs.
Because the SFS is an integral over the demography, the SFS from
a general measurable demography can be uniformly approximated by
a piecewise-constant demography with sufficiently many epochs.
Our results imply that it can be precisely obtained by a demography
with at most $2n-1$ epochs.

\section{Implications for Statistical Demographic Inference from Data}
\label{sec:ml_inference}

The SFS data used for demographic inference in population genomic studies are
noisy observations of the expected SFS from the underlying population
demography. Finite sequence lengths, ancestral/derived allele confounding, and
sequencing and variant calling errors are some common reasons for the empirical
SFS observed in sequencing studies differing substantially from the expected SFS
for the underlying demographic model. It is thus possible that the empirical SFS
observed in a sequencing study is not contained in the space of expected SFS
$\Xi_{n,\kappa_n}$ for any demographic model. Commonly used demographic
inference methods such as \dadi{} \citep{gutenkunst:2009}, \fastsimcoal{}
\citep{excoffier:2013}, and \fastNeutrino{} \citep{bhaskar:2015} perform
parametric demographic inference by searching for demographies which maximize
the likelihood of the observed SFS $\widehat{\sfsvec}'_n$. Under the widely used Poisson
Random Field model which assumes that the genomic sites being analyzed are
unlinked \citep{sawyer:1992}, maximizing the likelihood is equivalent to
minimizing the KL divergence between the empirical SFS and the expected SFS
under the parametric demographic model. Given an observed SFS $\widehat{\sfsvec}'_n$,
these algorithms traverse the interior of some user-specified space of
parametric population size functions such as $\Pi_k$, while computing the expected SFS $\widehat{\sfsvec}_n$
under the forward map $\chi(\vec{x},\vec{y})$ in each optimization iteration.
The optimization procedure either terminates and returns a demography $\eta$ in
the search space $\Pi_k$ whose expected SFS $\widehat{\sfsvec}_n(\eta)$ minimizes the KL
divergence $\text{KL}(\widehat{\sfsvec}'_n \,\|\,  \widehat{\sfsvec}_n(\eta))$ among all
demographies in $\Pi_k$, or it exhibits runaway behavior where some population
sizes or epochs diverge to infinity or go to 0 with successive optimization iterations.

Our geometric study of the $\widehat{\Xi}_{n,k}$ SFS manifold can clearly explain the
success and failure modes of these optimization algorithms. Suppose the
demographic search space is $\Pi_k$. When the observed SFS $\widehat{\sfsvec}'_n$ lies in the
interior of the $\widehat{\Xi}_{n,k}$ SFS manifold, this observed SFS is also exactly
equal to the expected SFS of some demographic model $\eta^* \in \Pi_k$, and
hence any of the above mentioned optimization algorithms, barring numerical
difficulties,\footnote{The package \dadi{} uses numerical methods to approximate
the solution to a diffusion PDE, while \fastsimcoal{} uses coalescent
simulations to estimate the expected SFS for a given demographic model. Hence,
these software packages might have numerical issues beyond the failure modes we
consider here. For this reason, we conduct our inference experiments in this section
using the \fastNeutrino{} package, which uses the analytic results of \citet{PK} for exact computation of the
expected SFS for piecewise-constant population size functions.} should be able
to find this demography $\eta^*$ whose expected SFS $\widehat{\sfsvec}_n(\eta^*)$ is
exactly equal to the observed data $\widehat{\sfsvec}'_n$ and $\text{KL}(\widehat{\sfsvec}'_n
\,\|\,  \widehat{\sfsvec}_n(\eta^*)) = 0$. On the other hand, if the noise in the
observed SFS causes it to lie outside the $\widehat{\Xi}_{n,k}$ SFS manifold, these
optimization algorithms will attempt to find the demography $\eta(t) \in \Pi_k$
which minimizes the projection under the KL divergence of the observed SFS
$\widehat{\sfsvec}'_n$ onto the $\widehat{\Xi}_{n,k}$ SFS manifold,
\begin{align}
\eta^{*} = \argmin_{\eta \in \Pi_k} ~\text{KL}(\widehat{\sfsvec}'_n  \,\|\, \widehat{\sfsvec}_n(\eta)). \label{eq:kl_proj}
\end{align}

The optimization problem in \eqref{eq:kl_proj} has a couple of issues. As
described in Sections \ref{sec:4k} and \ref{sec:nk}, the SFS manifold
$\widehat{\Xi}_{n,k}$ is a set where most of the boundary points are not contained in
$\widehat{\Xi}_{n,k}$. For example, the description of $\widehat{\Xi}_{4,2}$ in \sref{sec:4k} showed
that the only boundary point of $\widehat{\Xi}_{4,2}$ contained in $\widehat{\Xi}_{4,2}$ is the expected SFS
corresponding to the constant population size function, namely the point $(6/11,
3/11, 2/11)$. Since the points on the boundary of the $\widehat{\Xi}_{n,k}$ correspond to
different limiting regimes with the epoch durations and population sizes tending
to 0 or $\infty$, commonly used demographic inference algorithms that attempt to
solve the optimization problem in \eqref{eq:kl_proj} experience runaway behavior
when the observed SFS lies outside the $\widehat{\Xi}_{n,k}$ SFS manifold.
\fref{fig:constant_demography_projection} shows the SFS vectors of simulated
sequences (blue circles) under the coalescent with a constant population size, where most of
the simulated SFS fall outside the $\widehat{\Xi}_{4,2}$ manifold. 
We used \fastNeutrino{}
to fit two-epoch piecewise-constant demographies to these simulated SFS. The
observed SFS vectors which project onto the curved boundary of the upper convex
set are inferred to come from a demography where both the population size and
the duration of the recent epoch, $y_1$ and $t_1$ respectively, go to 0, while
for the observed SFS projecting onto the curved boundary of the lower convex
set, $y_1$ and $t_1$ diverge to infinity. In both cases, the location of the
projection along the curved boundaries is determined by the value of $x_1 = \exp(- t_1
/ y_1)$ where $x_1 \in (0, 1)$. On the other hand, for the observed SFS which
project onto the straight line boundaries of the upper and lower convex sets,
the inferred recent epoch durations go to 0 (upper convex set) or diverge to 
infinity (lower convex set) while the inferred
recent population size $y_1$ is a constant relative to the ancient population size $y_2$ and
the location of the projection along the boundaries is determined by $y_1/y_2$.

\begin{figure}[t]
\centerline{\includegraphics[width=0.95\textwidth]{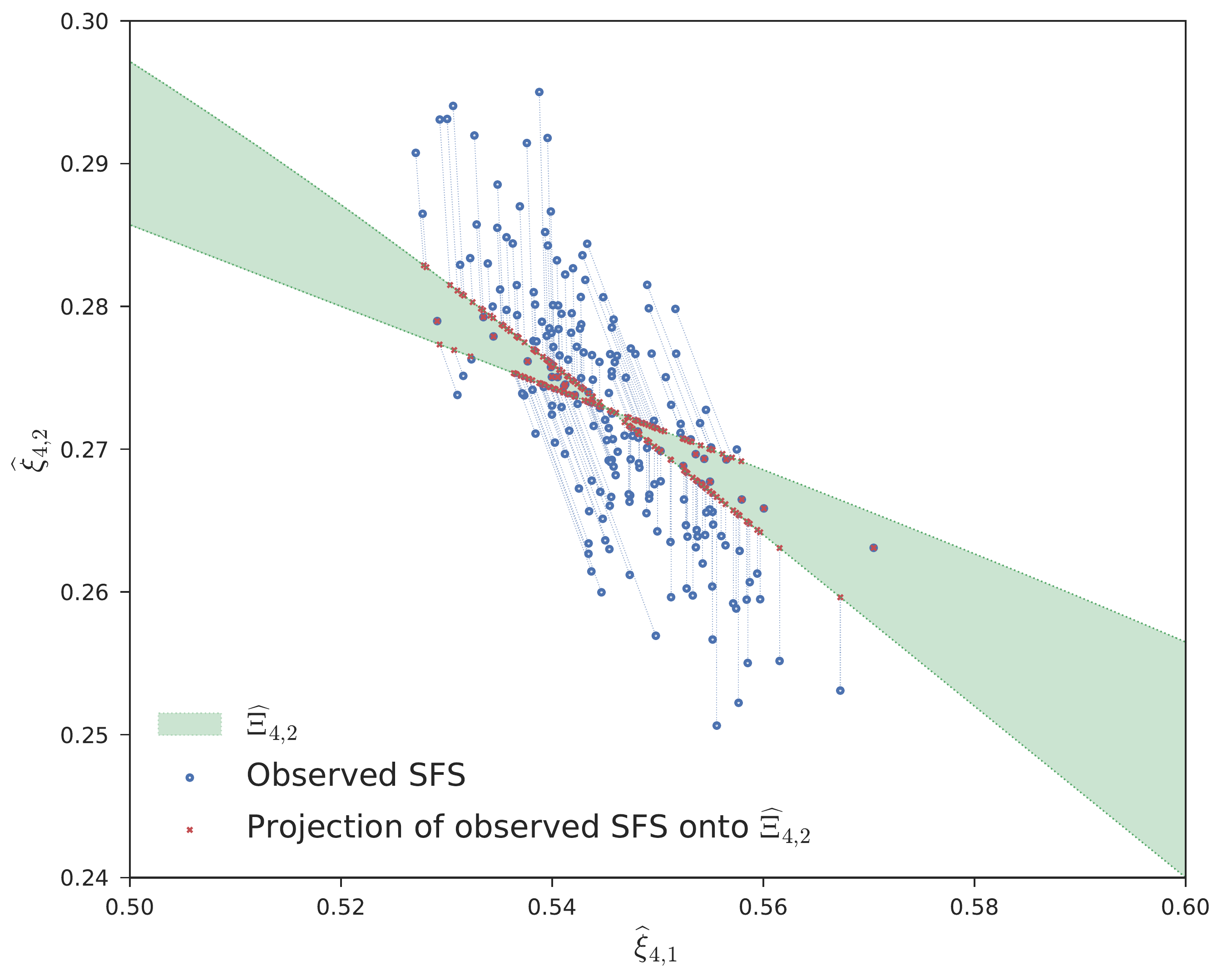}}
\caption{
Each blue circle is the observed SFS of $n=4$ haplotypes simulated using
\msprime{} \citep{kelleher:2016} under a constant population size coalescent
with recombination using realistic mutation and recombination rates of
$10^{-8}$ mutations and $2.2\times10^{-8}$ crossovers per basepair per
generation per haploid.  Each sequence has 1000 unlinked loci of length 10 kb
each, resulting in an average of 7{,}300 segregating sites. The red crosses
are the expected SFS at the two-epoch piecewise-constant demographies inferred
for these simulated SFS using \fastNeutrino{}; the red crosses are the
projections of the observed SFS onto the closure of $\widehat{\Xi}_{4,2}$ using the KL
divergence, with the dotted blue lines showing the correspondence between the
observed SFS and their projections.  For observed SFS lying in the
interior of $\widehat{\Xi}_{4,2}$, the observed SFS and their projections coincide,
while the observed SFS lying outside $\widehat{\Xi}_{4,2}$ project onto the boundaries of 
one of the two convex sets that form $\widehat{\Xi}_{4,2}$.
}
\label{fig:constant_demography_projection}
\end{figure}

A second more subtle issue arises from the fact that the $\widehat{\Xi}_{n,k}$ SFS
manifold being projected onto in \eqref{eq:kl_proj} may be a non-convex set, and
hence the solution to the optimization problem in \eqref{eq:kl_proj} may not be
unique.\footnote{If $\widehat{\Xi}_{n,k}$ is a convex set, the solution to
\eqref{eq:kl_proj} is unique due to the fact that the KL divergence is a convex
function of either argument. Namely, for $\widehat{\sfsvec}^{(1)}_n, \widehat{\sfsvec}^{(2)}_n \in
\widehat{\Xi}_{n,k}$, $\text{KL}(\widehat{\sfsvec}'_n \,\|\,  \lambda\widehat{\sfsvec}^{(1)}_n +
(1-\lambda)\widehat{\sfsvec}^{(2)}_n) \leq \lambda\,\text{KL}(\widehat{\sfsvec}'_n \,\|\, 
\widehat{\sfsvec}^{(1)}_n) + (1-\lambda)\,\text{KL}(\widehat{\sfsvec}'_n \,\|\,  \widehat{\sfsvec}^{(2)}_n)$, with
equality holding if and only if $\widehat{\sfsvec}^{(1)}_n = \widehat{\sfsvec}^{(2)}_n$.} For
example, we already observed in \sref{sec:4k} that the set $\widehat{\Xi}_{4,2}$ is non-convex 
and is given by the union of two convex sets. Hence, for $n=4$ and $k=2$, and
for some values of the observed SFS $\widehat{\sfsvec}'_4$, there could exist multiple
different demographic models $\eta_1, \eta_2 \in \Pi_2$, $\eta_1 \neq \eta_2$,
such that
\begin{align}
\inf_{\eta \in \Pi_2} \text{KL}(\widehat{\sfsvec}'_4 \,\|\,  \widehat{\sfsvec}_4(\eta)) = 
\text{KL}(\widehat{\sfsvec}'_4 \,\|\,  \widehat{\sfsvec}_4(\eta_1)) = 
\text{KL}(\widehat{\sfsvec}'_4 \,\|\,  \widehat{\sfsvec}_4(\eta_2)).
\end{align}
By algebraic considerations, the observed SFS $\widehat{\sfsvec}'_4$ which have such non-unique 
projections onto $\widehat{\Xi}_{4,2}$ form a set of measure zero among all possible
probability vectors on three elements, and hence such SFS are unlikely to be
encountered in real data. However, the existence of such SFS vectors
$\widehat{\sfsvec}'_4$ with non-unique projections implies that slight perturbations to
these vectors, say due to different quality control procedures for selecting the
set of genomic sites to analyze, could result in very different demographic
models being inferred due to the projection of the perturbed vector occurring
onto one or the other of the two convex sets composing $\widehat{\Xi}_{4,2}$. This is also
apparent in \fref{fig:constant_demography_projection}, where several pairs of observed SFS vectors that 
are very close to each other project onto the different convex sets forming $\widehat{\Xi}_{4,2}$.
As described in the previous paragraph, the boundaries of the upper and lower convex sets represent
various limiting regimes with either vanishingly small or arbitrarily large
recent population sizes and epoch durations, and this shows that even minor
perturbations to the SFS vector could yield qualitatively very different
inference results which cannot be reliably interpreted.

\section{Discussion}
\label{sec:discussion}

In this work, we characterized the manifold of expected SFS $\Xi_{n,k}$
generated by piecewise-constant population histories with $k$ epochs, while
giving a complete geometric description of this manifold for the sample size
$n=4$ and $k=2$ epochs. This special case is already rich enough to shed light
on the issues that practitioners can face when inferring population demographies
from SFS data using popular software programs.  While we demonstrated these
issues in \sref{sec:ml_inference} using the \fastNeutrino{} program, the issues
we point out are \emph{inherent} to the geometry of the SFS manifold and not
specific to any particular demographic inference software. Our simulations
showed that the demographic inference problem from SFS data can be fraught with
interpretability issues, due to the sensitivity of the inferred demographies to
small changes in the observed SFS data. These results can also be viewed as
complementary to recent pessimistic minimax bounds on the number of segregating
sites required to reliably infer ancient population size histories
\citep{terhorst:2015}.

Our investigation of piecewise-constant population histories also let us show a
general result that the expected SFS for a sample of size $n$ under \emph{any
population history} can also be generated by a piecewise-constant population
history with at most $2n-1$ epochs. This result could have potential
applications for developing non-parametric statistical tests of neutrality. Most
existing tests of neutrality using classical population genetic statistics such
as Tajima's $D$ \citep{tajima:1989} implicitly test the null hypothesis of
selective neutrality \emph{and} a constant effective population size
\citep{stajich:2004}. Exploiting our result characterizing the expected SFS of
samples of size $n$ under arbitrary population histories in terms of the
expected SFS under piecewise-constant population histories with at most
$\kappa_n$ epochs, we see that the KL divergence of an observed SFS $\sfsvec'_n$
to the expected SFS $\sfsvec_n(\eta^*)$ under the best fitting piecewise
constant population history $\eta^* \in \Pi_{\kappa_n}$ with at most $\kappa_n
\leq 2n-1$ epochs is also equal (up to a constant shift) to the negative log-likelihood 
of the observed SFS $\sfsvec'_n$ under the best fitting population
size history without any constraints on its form, assuming the commonly used
Poisson Random Field model where the sites being analyzed are unlinked. One can
then use the KL divergence inferred by existing parametric demographic inference
programs to create rejection regions for the null hypothesis of selective
neutrality without having to make any parametric assumption on the underlying
demography. Such an approach would also obviate the need for interpreting the
inferred demography itself, since the space of piecewise-constant population
histories is only being used to compute the best possible log-likelihood under
any single population demographic model. This approach could serve as an
alternative to recent works which first estimate a parametric demography using
genome-wide sites, and then perform a hypothesis test in each genomic region using
simulated distributions of SFS statistics like Tajima's $D$ under the inferred
demography \citep{rafajlovic:2014}. We leave the exploration of such tests for 
future work.

\section{Proofs}
\label{sec:proofs}

\subsection*{Proof of \propref{prop:2formulas}}
First, we reduce the integral expression for $c_m$ to a finite sum;
then we make appropriate manipulations until we arrive at the desired expressions.

Coalescence in the Wright-Fisher model is an inhomogeneous Poisson process
with parameter $\binom{m}{2}/\eta(t)$.
Therefore, the probability density of first coalescence at time $T$ is:
\begin{align*}
 \mathbb{P}(\text{No Coalescence in  } [0,T)) 
\mathbb{P}(\text{Coalescence at time  }T)
=  \exp\bigg[-\displaystyle \int_0^T \dfrac{\binom{m}{2}}{\eta(t)} \mathrm{d}t\bigg]
 \dfrac{\binom{m}{2}}{\eta(T)} \mathrm{dt}. 
\end{align*}
Let $R_{\eta}(t) = \displaystyle \int_0^T \dfrac{1}{\eta(t)} \mathrm{dt}$.
To compute the expected time to first coalescence, we have the integral:
\[\begin{array}{cclr} c_m & = & \displaystyle
\int_0^\infty t \cdot \dfrac{\binom{m}{2}}{\eta(t)}
\exp \left[- \binom{m}{2} R_{\eta(t)} \right] \mathrm{dt} &  \\[5mm]
& = & \displaystyle \int_0^\infty \exp \left[ - \binom{m}{2} R_{\eta}(t) \right] \mathrm{dt} & \text{(Integration by Parts)}
\end{array} \]
Substituting variables, $\tau = R_\eta(t)$, note that $ \mathrm{d} \tau =
\eta(R^{-1}(\tau))\mathrm{d}\tau$. Therefore, the integral becomes:
\[ c_m  = \displaystyle \int_0^\infty \tilde{\eta}(\tau)\exp \left[ - \binom{m}{2} \tau \right] \mathrm{dt}, \] where $\tilde{\eta}(\tau) = \eta(R^{-1}(\tau))$.

The population size $\eta(t)$ is a piecewise constant
function, whose value $\eta(t) = \eta_j$ if $t_{j-1} \leq t < t_{j}$.
As specified in the Proposition, $t_0 = 0$, $t_{k} = \infty$, and  $(y_1,\ldots,y_k)$
is the vector of population sizes.
Observe that $\tilde{\eta}(\tau)$ is also piecewise constant.
In particular,
\[ \tilde{\eta}(\tau) = \begin{cases} y_1, & 0 \leq \tau < \dfrac{t_1}{y_1}, \\ y_2, & \dfrac{t_1}{y_1} \leq \tau < \dfrac{t_1}{y_1} + \dfrac{t_2 - t_1}{y_2},  \\ \vdots & \vdots
\end{cases}
\]
Let $s_j = t_j - t_{j-1}$ for brevity. The resulting formula is:
\[\tilde{\eta}(\tau) = y_j,  \hspace{5mm} \text{   for     }  \sum_{k=1}^{j-1} \frac{s_k}{y_k} \leq \tau < \sum_{k=1}^j \frac{s_k}{y_k}. \]

 We turn the integral into a sum of integrals on
the constant epochs:
\[\begin{array}{lcl} c_m &  = & \displaystyle \int_0^\infty \tilde{y}(\tau)\exp \left[ - \binom{m}{2} \tau \right] \mathrm{d}\tau \\
&  = & \sum_{j = 1}^k \displaystyle \int_{ \sum^{j-1} s_l/y_l} ^{\sum^{j} s_l/y_l} y_j \exp \left[ - \binom{m}{2} \tau \right] \mathrm{d}\tau \\[5mm]
& = &  \sum_{j = 1}^k  y_j \left[ \dfrac{-1}{\binom{m}{2}} \exp \left[ - \binom{m}{2} \tau \right] \right]^ {\tau = \sum^{j} s_l/y_l}_{\tau = \sum^{j-1} s_l/y_l} \\[8mm]
& = & \dfrac{1}{\binom{m}{2}} \left\{\sum_{j=1}^k y_j \left( \exp \left[ - \binom{m}{2} \sum^j s_l / y_l \right] - \exp  \left[ - \binom{m}{2}\sum^{j-1} s_l /y_l \right]\right) \right\} \\[4mm]
& = & \dfrac{1}{\binom{m}{2}} \left\{\displaystyle\sum_{j=1}^k y_j \left(
\displaystyle\prod_{l = 1}^{j-1}\exp \left[-\binom{m}{2}s_l/y_l\right] \right) \left( 1 -
\exp \left[-\binom{m}{2}s_j/y_j\right]\right)\right\}.
\end{array}\]
We now make the substitution $x_j = \exp \left[- s_j / y_j\right] $.
Note that the old restriction $t_{j+1} > t_j > 0$ becomes
the new constraint $0 < x_j < 1$.
Our formula for the $c_m$ is now:
\[c_m  = \dfrac{1}{\binom{m}{2}}
\left[\sum_{j=1}^k y_j \left( \prod_{l = 1}^{j-1}x_l^{\binom{m}{2}} \right)
\left( 1 - x_j^{\binom{m}{2}}\right)\right].\]
Noting the linear form of this expression, we factor as a matrix
multiplication:
\[\left[\begin{array}{cccc} 1 & & &  \\ & \frac{1}{3} & & \\ & & \ddots & \\
& & & \frac{1}{\binom{n}{2}} \end{array} \right]
\left[\begin{array}{cccc} 1 & x_1 & \ldots & \prod_{i=1}^{k-1} x_i  \\
1 & x_1^3 & \cdots &  \prod_{i=1}^{k-1} x_i^3  \\
\vdots & \vdots & \ddots & \vdots \\
1 & x_1^{\binom{n}{2}} & \cdots &  \prod_{i=1}^{k-1} x_i^{\binom{n}{2}} \\
\end{array} \right]
\left[\begin{array}{ccccc} 
1 & 0 & 0 &  \cdots & 0 \\ 
-1 & 1 & 0 & \ddots & 0  \\
0 & -1 & 1 & \ddots & \vdots \\ 
\vdots  & \ddots  & \ddots & \ddots & 0 \\
0 & \cdots & 0 & - 1 & 1 \end{array} \right] \left[ \begin{array}{c}
y_1 \\ \vdots \\ y_k \end{array} \right].
\]
Combining the first three matrices yields \eqref{formula1}; combining the first two
and last two separately yields \eqref{formula2}. \qed

\subsection*{Proof of \propref{prop:trivial}}
We justify each equation in turn:
\begin{enumerate}
\item As mentioned in the introduction, this
is a classical result in population genetics,
and can be derived directly from \eqref{formula2}.
\item The inclusion $\C_{2,1} \subset \C_{2,k}$ is immediate,
so we need only show that any $a \in \C_{2,k}$ satisfies
$a > 0$. Using \eqref{formula1}, $a$ is written as a sum of
products of strictly positive numbers; so $\C_{2,k} \subset \C_{2,1}$.
\item First, we show that $\C_{3,2}$ is the interior of the open cone
spanned by $(1,0)$ and $(1,1)$. Fix $y_1 = a/(1-x_1)$ (for $a$ positive)
and consider $\chi(x_1,a/(1-x_1),y_2)$:
\medmuskip=1mu
\thinmuskip=1mu
\thickmuskip=1mu
\[ \chi(x_1,a/(1-x_1),y_2)= \left[ \begin{array}{c}
a + x_1 y_2 \\ \frac{1}{3} a(1 + x_1 + x_1^2) + \frac{1}{3}x_1^3 y_2
\end{array} \right] =  a \left[ \begin{array}{c}
1 \\ \frac{1}{3}(1 + x_1 + x_1^2)\end{array} \right] + x_1 y_2 \left[ \begin{array}{c}
1 \\ \frac{1}{3}x_1^2
\end{array} \right] .\]
When $x_1 \to 0$, the second vector approaches $(1,0)$; when $x_1 \to 1$, the
first vector approaches $(1,1)$.
The vectors are in the interior of that cone for all other
permissible values of $x_1$ and $y_2$.
To show that $\C_{3,k} = \C_{3,2}$, note that for
larger values of $k$, the same cone of
vectors are produced. In particular,
 $\chi(x_1,\ldots,x_{k-1},y_1,\ldots,y_k)$ yields
\medmuskip=-1mu
\thinmuskip=-1mu
\thickmuskip=0mu
\nulldelimiterspace=0pt
\scriptspace=1pt
\[
 \mathlarger{\sum_{j=1}^{k-1}} \left\{ y_j\left( \prod_{i = 1}^{j-1} x_i \right) (1-x_j)
\left[ \begin{array}{c}
1 \\ \frac{1}{3}( \prod_{i = 1}^{j-1} x_i^2) (1 + x_j + x_j^2)
\end{array} \right]\right\} +  y_k \left( \prod_{i = 1}^{k-1} x_i\right) \left[ \begin{array}{c}
1 \\ \frac{1}{3} ( \prod_{i = 1}^{k-1} x_i^2)
\end{array} \right] .\]
Clearly, the second coordinate of all vectors is bounded between $0$ and $1$. \qedhere
\end{enumerate}

\subsection*{Proof of \propref{prop:c42}}

First we observe that $\gamma_1$ and $\gamma_2$ are normalizations
of the curves defined by parameterizations $(t, \frac{1}{3}t^3,\frac{1}{6}t^6)$ and
$(1-t, \frac{1}{3}(1-t^3),\frac{1}{6}(1-t^6))$ where $t$ is constrained to the open
interval $(0,1)$.

Now we claim that the definition in terms of the map $\chi(x,y)$ is
equivalent to the definition in terms of these two curves.
We can use the first formulation of $\chi$ to prove this:

\begin{align*}
\chi(x_1,y_1,y_2) & =  y_1 
\left[ \begin{array}{c}
1 -x_1 \\ (1 -x_1^3)/3 \\ (1 - x_1^6)/6 \end{array} \right] 
+ y_2 \left[ \begin{array}{c} x_1 \\ x_1^3/3 \\ x_1^6/6 \end{array} \right] &
\\
& =  y_1  \left[ \begin{array}{c} 1 \\ \frac{1}{3} \\ \frac{1}{6} \end{array} \right]
+(y_2-y_1) \left[ \begin{array}{c} x_1 \\ \frac{1}{3}x_1^3 \\ \frac{1}{6}x_1^6 \end{array} \right]
 = (y_1 - y_2) \left[ \begin{array}{c}
1 -x_1 \\ \frac{1}{3}(1 -x_1^3) \\ \frac{1}{6}(1 - x_1^6) \end{array} \right]  +y_2 \left[ \begin{array}{c} 1 \\ \frac{1}{3} \\ \frac{1}{6} \end{array} \right] .
\end{align*}

When $y_2 = y_1$, the image is the point $(2/3,2/9,1/9) = X$ as
stated. When $y_2 > y_1$, we can use the left-hand expression to view
the image as a point on the line segment between $\C_{4,1}$ and the
curve $(t, t^3/3,t^6/6)$. When $y_2 < y_1$, the right-hand expression can
be used to write the image as a point on the line segment between $X$
and $(1-t,(1-t^3)/3,(1-t^6)/6)$. This means that the image of $\chi$ is
contained in the regions and point specified.

To show that the reverse inclusion holds, we fix a point $P$ in the interior
of the convex hull of $\gamma_1$. By convexity, the line segment
from $X$ to $P$ is contained in the region; continue in the direction
$P-X$ until the line intersects the curve. This must occur because all
points in the region are further from the bounding line than $X$.
The point of intersection $q$ is specified as $q = \gamma_1(\tau)$
for some $\tau \in (0,1)$. By convexity,
there exists some $\rho$ such that $\rho \:\:
\C_{4,1} + (1-\rho)q = P$. Fixing $x_1 = \tau$,
$y_1 = \rho$ and $y_2 = 1$, shows that
$P$ is in the image of $\chi$. The same argument
holds with slight variation for $\gamma_2$.

\subsection*{Proof of \propref{prop:c43}}

The strategy to prove the equality of $\C_{4,3}$ and the cone
over $\{ t,t^3, t^6\}$ comes in two steps:
\begin{enumerate} \item Show that the columns of $M_1(4,k)$ are
always contained in the region $R$ whose boundary is $\gamma_1 \cup \gamma_2 \cup
\gamma_3$.

\item Divide the convex hull of $R$ into two regions and show that
each of these regions are included in $\widehat{\C}_{4,3}$.
\end{enumerate}

First we demonstrate that the regions maps precisely into $R$. We have
already shown in the main text of the document that the boundaries of
$(0,1) \times (0,1)$ map to the boundaries of $R$ under the mapping defined
by $(x_1,x_2)\mapsto (x_1(1-x_2),x_1^3(1-x_2^3)/3,x_1^3(1-x_2^3)/6) \times 1/S$,
where $S$ is the sum of the coordinates. We compute
the Jacobian of this map explicitly in {\tt Macaulay2} \citep{M2}. The result
is: \[-\frac{1}{6S^3}x_1^9(x_2-1)^4(x_2^2+x_2+1)(x_2^2+3x_2+1). \]
Plainly, this is nowhere zero in our domain. The inverse function theorem
then implies that the interior is contained in the image of the boundaries.
This accomplishes Step 1 of our proof.

For Step 2, we divide the image into two regions: 
\begin{enumerate}
\item The triangle defined by vertices $(1,0,0)$,
$(2/3,2/9,1/9)$ and $(1/3,1/3,1/3)$, including the
two edges $[(1/3,1/3,1/3),(2/3,2/9,1/9)]$ and
$[(2/3,2/9,1/9),(1,0,0)]$.
\item The remainder of the convex hull of $R$ --
explicitly, the interior of the region bounded by $\gamma_3$ and the line
segment $[(1/3,1/3,1/3),(1,0,0)]$.
\end{enumerate}

To show that the triangle is included, let $x_2 = \epsilon \approx
0$, and let $x_1$ vary. Then the third column sits arbitrarily
close to $(1,0,0)$ and the first column traces out $\gamma_2$.
Set $y_2 \approx 0$ and toggle $y_1$ and $y_3$, to obtain the full
span, including the interior of the triangle, and the line segment
$[(1/3,1/3,1/3),(2/3,2/9,1/9)]$.
Set $x_1 = 1-\epsilon$, and the first column sits at $(1/3,1/3,1/3)$
while the third column traces out $\gamma_1$. This catches the missing
line segment.

For the remainder of the convex hull, fix a point $P$ in this region.
This point lies on a line segment between
$(2/3,2/9,1/9)$ and some point $Q$ in $\gamma_3$. Suppose it is
equal to $\rho \cdot (2/3,2/9,1/9) + (1-\rho)\cdot Q$.
Set $x_2 = 1 - \epsilon \approx 1$.
We can choose $\epsilon$ and $x_1$ so that the second column
is arbitrarily close to $P$. Furthermore, observe that the first column
is approximately equal to the point on
$\gamma_2$ corresponding to $x_1$ and the third column is approximately
the point on $\gamma_1$ corresponding to $x_1$.
Choosing $y_1 = y_3 = \rho $ and $y_2 = 1 - \rho$ points us to
\[\rho \cdot \left(\left( \begin{array}{c}| \\ \gamma_1(x_1) \\| \end{array} \right)
+ \left( \begin{array}{c} | \\ \gamma_2(x_1) \\| \end{array} \right) \right)
+ (1 - \rho)\cdot \left( \begin{array}{c}| \\ \gamma_3(x_1) \\| \end{array} \right)
= \rho \cdot \left( \begin{array}{c} 2/3 \\ 2/9 \\ 1/9 \end{array} \right)   + (1-\rho) \cdot Q = P.\]

\subsection*{Proof of \propref{prop:sfs}}

This is a direct application of the linear map $W_4$,
computed as in \cite{PK}:
\[
W_4 = \left(\begin{array}{ccc} 6/5 & 2 & 4/5 \\ 6/5 & 0 & -6/5 \\ 6/5 & -2 & 4/5 \\
\end{array} \right).
\]

\subsection*{Proof of \propref{prop:cnk}}

In order to prove the result about dimension, we show that
$\C_{n,k}$ is a relatively open subset of a certain algebraic variety.
Because the relevant operations are native to projective geometry, we
transport our objects of interest in the obvious way to projective space.
The same scaling properties that allow us to focus on the simplex also
lead to good behavior in projective space.

\begin{lemma}
For $k \geq 2$ , the Zariski closure of $\C_{n,k}$ is the affine cone
over $\join(\sigma_{k-2}(C_n,p_n))$, where:
\begin{enumerate}
\item $C_n$ is the projective curve defined by mapping $[s:t]$ to
\[ C_n = \left[{\binom{2}{2}}^{-1}  \spex{s}{4pt}{\binom{n}{2} - \binom{2}{2}} \spex{t}{4pt}{\binom{2}{2}}:
{\binom{3}{2}}^{-1} \spex{s}{4pt}{\binom{n}{2} - \binom{3}{2}} \spex{t}{4pt}{\binom{3}{2}}:
\cdots:
{\binom{n}{2}}^{-1}  \spex{t}{4pt}{\binom{n}{2}}\right],
\]
\item $p_n$  is the projective point $\Big[1: \frac{1}{3}: \frac{1}{6}: \cdots: \frac{1}{\binom{n}{2}}\Big]$,
\item $\mathcal{J}$  denotes the join of algebraic varieties, and
\item $\sigma_{i}(\cdot)$ denotes the $i$-th secant variety.
Following \citet{Harris}, the $i$-th secant variety is the union
of $i$-dimensional planes generated by $i+1$ points in the variety.
\end{enumerate}
\label{lem:join}
\end{lemma}

\begin{proof}[Proof of Lemma~\ref{lem:join}]
The variety $\join(\sigma_{k-2}(C_n),p_n))$ is the image of the following map:
\[\psi(\vec s, \vec t, \vec \lambda) = \left( \begin{array}{cccc} 1 & s_1^{\binom{n}{2} - 1}
 t_1 & \cdots &
 s_{k-1}^{\binom{n}{2} - 1}
t_{k-1} \\[1mm]
\frac{1}{3} & \frac{1}{3} s_1^{\binom{n}{2} - 3} t_1^3  & \cdots & \frac{1}{3}  s_{k-1}^{\binom{n}{2} - 3} t_{k-1}^3 \\
\vdots & \vdots & \ddots & \vdots \\[1mm]
\frac{1}{\binom{n}{2}} & \frac{1}{\binom{n}{2}} t_1^{\binom{n}{2}} & \cdots & \frac{1}{\binom{n}{2}} t_{k-1}^{\binom{n}{2}} \end{array} \right)
\left( \begin{array}{c} \lambda_0 \\ \lambda_1 \\ \vdots
\\ \lambda_{k-1} \end{array}\right),
\]
where $s_i$ and $t_i$ are not simultaneously zero, and $\lambda$ is unrestricted.

Define the map $\phi: \RR^{2k-1} \to (\PP^1)^{k-1} \times \RR^{k}$ sending
\[ (x_1,\ldots,x_{k-1},y_1,\ldots,y_k) \mapsto
\left([1:x_1],[1:x_1x_2],\ldots,[1:\prod_{i=1}^{k-1} x_i],
y_1,y_1 + y_2, \ldots,\sum_{i=1}^k y_i \right)\]
We can recast the expression in \eqref{formula2}
as the composition $\psi \circ \phi$.
Based on this formulation, the set $\C_{n,k}$ is clearly contained
in $\join(\sigma_{k-2}(C_n),p)$.
To demonstrate the equality of the Zariski closures,
we only need to show that the dimensions match and
that the variety is irreducible. Both joins and secants have
the property that irreducible inputs yield irreducible outputs, so
the variety of interest is irreducible. The image of
$\phi$ is open in $(\PP^{1})^{k-1} \times \PP^{k-2}$, and the
map $\psi$ has deficient rank on a set of positive codimension.
Therefore, the composition of $\psi \circ \phi$ has full dimension.
This proves the Lemma.
\end{proof}

The $i$-th secant variety of an irreducible nondegenerate curve in $\PP^n$
has projective dimension given by $\min(2i+1,n)$
\citep[Exercise 16.16]{Harris}. The curve $C_n$ is a toric transformation
of a coordinate projection of the rational normal curve.
The rational normal curve is nondegenerate, and both of
these operations preserve that property.
This means our secant variety has projective dimension
$\min(2(k-2)+1,n-2) = \min(2k-3,n-2)$.
The join with a point adds 1 to the dimension of the variety,
while the operation of passing to the affine cone adds 1 to the dimension
of the variety {\bf and} the ambient space. However, normalizing to the
$(n-2)$-simplex subtracts 1 from both variety and ambient space again.
This means that $\dim \widehat{\C}_{n,k} = \min(2k-2,n-2)$,
assuming that $k \geq 2$.

\subsection*{Proof of upper bound in \thmref{thm:kappa_bounds}}

Suppose a point ${\bf c}$ is in $\C_{n,q}$. By definition, this implies that
there is a point $(x_1,\ldots,x_{q-1},y_1,\ldots,y_q)$ such that \eqref{formula1} yields

\begin{small}
\medmuskip=-1mu 
\thinmuskip=-1mu
\thickmuskip=-1mu
\nulldelimiterspace=0pt
\scriptspace=0pt
\[\left[
\begin{array}{ccccc}
1-x_1 & x_1(1-x_2) & \cdots & \prod_{i=1}^{q-2} x_i(1- x_{q-1}) & \prod_{i=1}^{q-1} x_i \\
(1/3)(1-x_1^3) & (1/3)x_1^3(1-x_2^3) & \cdots & (1/3)\prod_{i=1}^{q-2} x_i^3(1- x_{q-1}^3) & (1/3)\prod_{i=1}^{q-1} x_i^3 \\
\vdots & \vdots & \ddots & \vdots & \vdots\\
\binom{n}{2}^{-1}(1-\spex{x_1}{4pt}{\binom{n}{2}}) &
\binom{n}{2}^{-1}\spex{x_1}{4pt}{\binom{n}{2}}(1-\spex{x_1}{4pt}{\binom{n}{2}}) &  \cdots &
\binom{n}{2}^{-1}\prod_{i=1}^{q-2}\spex{x_i}{4pt}{\binom{n}{2}} (1-\spex{x_{q-1}}{4pt}{\binom{n}{2}}) &
\binom{n}{2}^{-1}\prod_{i=1}^{q-1} \spex{x_i}{4pt}{\binom{n}{2}} \\
\end{array} \right]
\left[\begin{array}{c}
y_1 \\ \vdots \\ y_q
\end{array}
\right].
\]
\end{small}

\noindent Since the point ${\bf c}$ is in the cone
over the $q$ columns of the matrix,
Carath\'{e}odory's Theorem implies
that it is also in the cone over some $n-1$ of the
columns. Therefore we can replace the vector
$y_1,\ldots,y_q$ with $y'_1,\ldots, y'_q$
so that all but $n-1$ (or fewer) are zero.

Passing to the expression in \eqref{formula2}, this gives us:
\begin{small}
\[\left[
\begin{array}{ccccc}
1 & x_1 & \cdots & \prod_{i=1}^{q-2} x_i & \prod_{i=1}^{q-1} x_i \\
(1/3) & (1/3)x_1^3 & \cdots & (1/3)\prod_{i=1}^{q-2} x_i^3& (1/3)\prod_{i=1}^{q-1} x_i^3 \\
\vdots & \vdots & \ddots & \vdots & \vdots\\
\binom{n}{2}^{-1} &
\binom{n}{2}^{-1}\spex{x_1}{4pt}{\binom{n}{2}} &  \cdots &
\binom{n}{2}^{-1}\prod_{i=1}^{q-2}\spex{x_i}{4pt}{\binom{n}{2}} &
\binom{n}{2}^{-1}\prod_{i=1}^{q-1} \spex{x_i}{4pt}{\binom{n}{2}} \\
\end{array} \right]
\left[\begin{array}{c}
y_1'  \\
y_2' - y_1' \\
\vdots \\ y_q' - y_{q-1}'
\end{array}
\right].
\]
\end{small}

\noindent
Since at most $n-1$ of the $y_i'$ are nonzero, at most
$2n-2$ of the indices of the vector at right are nonzero.
We can delete the columns of the $X$ matrix corresponding to zero
entries except the first column.
A new sequence $(x_1',\ldots,x_{2n-2}')$ can then be obtained
from the ratio between the first entries in adjacent columns.
The new sequence $y_1'',\ldots,y_{2n-1}''$
can be obtained by taking the sequence of partial sums of the vector.

\section*{Acknowledgments} 
We thank the Simons Institute for the Theory of Computing, where some of this
work  was carried out while the authors were participating in the ``Evolutionary Biology
and the Theory of Computing'' program.
This research is supported in part by a Math+X Research Grant, 
an NSF grant DMS-1149312 (CAREER), an NIH grant R01-GM109454, and a Packard Fellowship for Science and Engineering.
YSS is a Chan Zuckerberg Biohub investigator.

\bibliographystyle{myplainnat}

\bibliography{sfs}

\end{document}